\documentclass[11pt]{article}

\usepackage{amsmath,amsgen}
\usepackage{amstext,amssymb,amsfonts,latexsym,mathrsfs}
\usepackage{theorem} \usepackage{pifont}\usepackage[dvips]{graphics,epsfig}

 \newcommand{\bs}{\bigskip}
 \newcommand{\ms}{\medskip}
 \newcommand{\n}{\noindent}
 
 \newcommand{\hs}[1]{\hspace*{ #1 mm}}

\def\bbox{\vrule height6pt width6pt depth1pt}
\theoremstyle{plain}
\theoremheaderfont{\bfseries}
\newtheorem{theorem}{Theorem}[section]
\newtheorem{lemma}[theorem]{Lemma}
\newtheorem{proposition}[theorem]{Proposition}

{\theorembodyfont{\rmfamily}
     }
{\theorembodyfont{\rmfamily} }
{\theorembodyfont{\rmfamily} }

\newenvironment{proof}{\par \noindent
            {\bf Proof. \hs{2}}}{\hfill$\Box$ \vspace*{3mm}}

\newenvironment{proofof}[1]{\vspace*{5mm} \par \noindent
         {\bf Proof of #1.\hs{2}}}{\hfill$\Box$ \vspace*{3mm}}

 \newcommand{\ignore}[1]{}

 \newcommand{\ket}[1]{| #1 \rangle}

\begin{document}
\pagestyle{plain}
\begin{center}
{\Large {\bf The Efficiency of Quantum Identity Testing of Multiple States}}
\bs\\

\begin{center}
{\sc Masaru Kada}$^1$ 
\hspace{5mm} 
{\sc Harumichi Nishimura}$^1$
\footnote{Supported in part by Scientific Research Grant, 
Ministry of Japan, 19700011.} 
\hspace{5mm} 
{\sc Tomoyuki Yamakami}$^2$
\footnote{Supported in part by grants from the Mazda 
Foundation and from the Grant-in-Aid for Scientific Research of Japan.}
\end{center}

$^1${School of Science, Osaka Prefecture University}

{\tt $\{$kada,hnishimura$\}$@mi.s.osakafu-u.ac.jp}

$^2${School of Computer Science and Engineering, University of Aizu}

{\tt yamakami@u-aizu.ac.jp}
\end{center}

\bs
\begin{abstract}
\ms
\n We examine two quantum operations, the Permutation Test and the Circle Test, which test the identity of $n$ quantum states. These operations naturally extend the well-studied Swap Test on two quantum states. We first show the optimality of the Permutation Test for any input size $n$ as well as the optimality of the Circle Test for three input states. In particular, when $n=3$, we present a semi-classical protocol, incorporated with the Swap Test, which approximates the Circle Test efficiently. Furthermore, we show that, with help of classical preprocessing, a single use of the Circle Test can approximate the Permutation Test efficiently for an arbitrary input size $n$. 
    

\end{abstract}

\section{\label{sec1}Introduction}

When we manipulate quantum information,  one of the 
fundamental operations is to compare two or more pieces of quantum information. 
In particular, we wish to test whether two quantum states are identical or nearly orthogonal to each other. A standard quantum operation to test the identity of two quantum states is the {\em (Controlled) Swap Test}, which ``conditionally'' swaps the two quantum states and obtains an answer by  measuring its controlled qubit. The Swap Test finds a direct application to, for instance, the {\em fingerprinting} protocol of Buhrman, Cleve, Watrous, and de Wolf\ \cite{BCWW01}. They considered the following three-party communication game (known as 
{\em simultaneous message passing model} in communication complexity \cite{KN97}). Two parties, Alice and Bob, hold $m$-bit inputs $x$ and $y$, respectively, and Referee wishes to calculate a desired value $f(x,y)$ correctly with high probability,
based solely on the messages received from Alice and Bob, who are prohibited to communicate with each other.

For instance, the equality function $EQ$ ({\it i.e.}, $EQ(x,y)=1$ if $x=y$ and $0$ otherwise)
requires, by a quantum operation of Buhrman {\it et al.}\ \cite{BCWW01}, 
Alice and Bob to send quantum information of $O(\log m)$ qubits to Referee, 
who applies the Swap Test over the received quantum states to test whether $x=y$ (and thus computes $EQ(x,y)$). 
In stark comparison, Alice and Bob should send $\Omega(\sqrt{m})$ bits of classical information to Referee \cite{NS96,BK97} (this bound turns out to be tight \cite{Amb96}) to compute the equality function. The usefulness of the Swap Test in the above protocol of Buhrman {\it et al.}\ stems from the fact that two quantum states received from Alice and Bob are either identical (when $x=y$) or nearly orthogonal (when $x\neq y$). 

Besides \cite{BCWW01}, the Swap Test has been a key player in various fingerprinting protocols in, {\it e.g.}, \cite{AS04,Bea04,GKW06,SWS07,Yao03}.
Moreover, the Swap Test has been used in various physical and computational settings, which include stabilization of quantum computation \cite{Bar97}, quantum estimation \cite{Eke02}, 
quantum Merlin-Arthur games \cite{KMY03}, and black-box group problems \cite{Fri03}. 
Nevertheless, the Swap Test handles only two quantum states. How can we test the identity of more than two quantum states?

This paper examines two natural generalizations of the Swap Test, referred to as the {\em Permutation Test} 
and the {\em Circle Test}, which turn out to be useful tools in testing the identity of three or more quantum states. 
Instead of swapping two states in the Swap Test, the Permutation Test ``conditionally'' permutes $n$ input states by applying, in superposition, all possible permutations over $n$ elements. The Circle Test is a simpler form of the Permutation Test using only multiple applications of a single permutation. (For their formal definitions, see Section \ref{sec2}.)
In a slightly different context, the Permutation Test 
can be used to amplify 
the success probability of the aforementioned quantum protocol 
for $EQ$ \cite{BCWW01}. 
In this paper, our focal point 
is the following problem of testing the identity of $n$ quantum states in a state space ${\cal H}$, provided that these states are either identical or mutually orthogonal, for simplicity of our argument.

\medskip

{\bf Quantum State Identity Problem ($\mathrm{QSI}_n$)} 

{\bf Input:} $n$ quantum states $(|\psi_1\rangle,|\psi_2\rangle,\ldots,|\psi_n\rangle)$ 
in a state space ${\cal H}$. 

{\bf Promise:} Any pair of the $n$ quantum states is equal or orthogonal. 

{\bf Output:} YES if all $n$ states are identical; NO otherwise. 

\medskip

Over two input states, the Swap Test can solve the above identity problem $\mathrm{QSI}_2$ by outputting ``EQUAL'' on any ``YES'' instance with certainty ({\em completeness error probability} $0$) 
and outputting ``NOT EQUAL'' on any ``NO'' instance with probability exactly $1/2$ 
({\em soundness error probability} $1/2$). Under the so-called {\em one-sided error requirement}, in which the completeness error probability should be $0$, the Swap Test is known to be 
an optimal quantum operation for the identity problem $\mathrm{QSI}_2$. 
This fact was implicitly proven in 2001 by Kobayashi, Matsumoto, and Yamakami \cite{KMY01} (see also \cite{Bea04}). 
In Section 2, we show the optimality of the Circle Test as well as the Permutation Test under the same one-sided error requirement; more precisely, the Circle Test is an optimal operation for the problem $\mathrm{QSI}_3$, 
and the Permutation Test is optimal for $\mathrm{QSI}_n$ for an arbitrary input size $n\geq2$. 

Subsequently, we present efficient approximations of the Circle Test and the Permutation Test using ``semi-classical'' protocols involving the Swap Test and the Circle Test, respectively. As a direct consequence, 
these approximations help us build a concise quantum circuit that solves the problem $\mathrm{QSI}_n$ efficiently,  
because a quantum circuit that implements the Swap Test (resp. the Circle Test) is significantly more concise than any quantum circuit for the Circle Test (resp. the Permutation Test). 
In Section 3, we show how a certain sequential application of the Swap Test efficiently approximates the Circle Test for $\mathrm{QSI}_3$. Such an operation gives an optimal approximation procedure. 
In Section 4, we show that, with help of classical preprocessing, a single application of the Circle Test 
can approximate the Permutation Test for $\mathrm{QSI}_n$ with efficiency, which is one-sided error 
and has optimal soundness error probability up to a multiplicative factor of smaller than $2$. 
We conclude in Section \ref{sec5} with an extension of our results and also  a suggestion of future directions. 

\section{\label{sec2}The Permutation Test and the Circle Test}

Besides the well-studied Swap Test, we introduce two useful tests, called the {\em Permutation Test} and the {\em Circle Test}, which intend to solve our quantum state identity problem  $\mathrm{QSI}_n$ on $n$ input states taken 
from a state space ${\cal H}$. 
We begin with the formal definition of the Permutation Test on $n$ quantum states 
$(|\psi_1\rangle,|\psi_2\rangle,\ldots,|\psi_n\rangle)\in{\cal H}^{\otimes n}$.
For our notational convenience, let $\sigma = \{\sigma_0,\sigma_1,\ldots,\sigma_{n!-1}\}$ denote the set of all $n!$ permutations 
over the integer set $[n]:=\{1,2,\ldots,n\}$; namely, for each $i\in[n]$, $\sigma_i$ denotes the $i$-th element 
of the symmetric group $S_n$ (in a certain fixed order). 
Notice that the Swap Test is in fact the Permutation 
Test on two quantum states.

\medskip
{{\bf  Permutation Test}

Input: $n$ quantum states $(|\psi_1\rangle,|\psi_2\rangle,\ldots,|\psi_n\rangle)$ in a state space ${\cal H}$.  
\begin{enumerate}
\item Start with the quantum state $|0\rangle\otimes |\psi_1\rangle\otimes\cdots\otimes|\psi_n\rangle$,  
where $|0\rangle$ (often called the first register) denotes the ground state in the $n!$-dimensional 
state space. 

\item Apply the quantum Fourier transform $F_{n!}$ over $n!$ elements to the first register. 

\item Apply a controlled-$\sigma$ operation; that is, if the first register contains index $i\in\{0,1,\ldots,n!-1\}$, 
transform $|\psi_1\rangle\otimes\cdots\otimes|\psi_n\rangle$ 
to $|\psi_{\sigma_i(1)}\rangle\otimes\cdots\otimes |\psi_{\sigma_i(n)}\rangle$. 

\item Apply the inverse quantum Fourier transform $(F_{n!})^{-1}$ to the first register.

\item Measure the first register in the computational basis. If $0$ is observed, output EQUAL; 
otherwise, output NOT EQUAL. 
\end{enumerate}

The Circle Test is a simple form of the Permutation Test, defined by multiple applications of a single permutation, 
denoted $\sigma_c$, where $\sigma_c$ is the permutation on $[n]$ of the following form: $\sigma_c(n)=1$ and $\sigma_c(i)=i+1$ for any index $i\in[n-1]$. The notation $\sigma^{j}_c(i)$ means the result of the $j$ applications of $\sigma_c$ to $i$.

Our motivation of introducing the Circle Test is to provide a tool in building a ``concise'' quantum circuit that solves $\mathrm{QSI}_n$ efficiently. Consider a quantum circuit that implements the Permutation Test for the problem 
$\mathrm{QSI}_n$. Since the Permutation Test involves the quantum Fourier transform $F_{n!}$ over $n!$ elements, a straightforward decomposition of such a transform gives a large-size quantum circuit for $\mathrm{QSI}_n$. It is therefore better to use a simpler quantum test (than the Permutation Test) to solve the problem $\mathrm{QSI}_n$ with efficiency.

\medskip
{\bf Circle Test} 

Input: $n$ quantum states $(|\psi_1\rangle,|\psi_2\rangle,\ldots,|\psi_n\rangle)$ in a state space ${\cal H}$.  
\begin{enumerate}
\item Start with the quantum state $|0\rangle\otimes |\psi_1\rangle\otimes\cdots\otimes|\psi_n\rangle$ 
where $|0\rangle$ (often called the first register) denotes the ground state in the $n$-dimensional 
state space.  

\item Apply the quantum Fourier transform $F_{n}$ to the first register. 

\item Apply a controlled-$\sigma_c$ operation; namely, when the first register contains $i\in\{0,1,\ldots,n-1\}$, 
transform $|\psi_1\rangle\otimes\cdots\otimes|\psi_n\rangle$ to $|\psi_{\sigma_c^i(1)}\rangle\otimes\cdots\otimes|\psi_{\sigma_c^i(n)}\rangle$. 

\item Apply the inverse quantum Fourier transform $(F_n)^{-1}$ 
to the first register.

\item Measure the first register in the computational basis. If $0$ is observed, output EQUAL; 
otherwise, output NOT EQUAL. 
\end{enumerate}

In particular, when $n=2$, the Permutation Test as well as the Circle Test coincide with the Swap Test. 
For later analysis, we show how to calculate the probabilities that our new tests on $n$ input states 
output EQUAL. 

\begin{lemma}\label{circle_test}
Given $n$ input states $(|\psi_1\rangle,\ldots,|\psi_n\rangle)\in {\cal H}^{\otimes n}$, 
the probabilities that the Permutation Test and the Circle Test output EQUAL are, respectively,  
\begin{equation}\label{perm_formula}
\frac{1}{n!}\sum_{k=0}^{n!-1}\prod_{m=1}^{n} \langle\psi_m|\psi_{\sigma_k(m)}\rangle
\hspace{1cm}
and
\hspace{1cm}
\frac{1}{n}\sum_{k=0}^{n-1}\prod_{m=1}^{n} \langle\psi_m|\psi_{\sigma_c^k(m)}\rangle.
\end{equation}
\end{lemma}

\begin{proof}
We show the lemma only for the Circle Test, because the case of the Permutation Test can be similarly proven. Let $(|\psi_1\rangle,\ldots,|\psi_n\rangle)\in {\cal H}^{\otimes n}$ be our $n$ input states. 
The Circle Test outputs EQUAL on these input states with probability 
exactly $\|\sum_{i=0}^{n-1}|\psi_{\sigma_c^i(1)}\rangle\cdots|\psi_{\sigma_c^i(n)}\rangle\|^2/n^2$, which can be further simplified as
\begin{eqnarray*}
\hs{-15}\lefteqn{\frac{1}{n^2}\sum_{i=0}^{n-1}\sum_{j=0}^{n-1}
\langle\psi_{\sigma_c^i(1)}|\psi_{\sigma_c^j(1)}\rangle\cdots\langle\psi_{\sigma_c^i(n)}|\psi_{\sigma_c^j(n)}\rangle}\hs{5} \\
\hs{-10}&=& \frac{1}{n^2}\sum_{k=0}^{n-1}\sum_{i=0}^{n-1}\prod_{m=1}^n\langle\psi_{\sigma_c^i(m)}|\psi_{\sigma_c^{i+k}(m)}\rangle 
\;\;=\;\; \frac{1}{n^2}\sum_{k=0}^{n-1}\sum_{i=0}^{n-1}\prod_{m=\sigma_c^{-i}(1)}^{\sigma_c^{-i}(n)}
\langle\psi_{\sigma_c^i(m)}|\psi_{\sigma_c^{i+k}(m)}\rangle.
\end{eqnarray*}
Clearly, the last expression equals $\frac{1}{n}\sum_{k=0}^{n-1}\prod_{m=1}^n\langle\psi_m|\psi_{\sigma_c^k(m)}\rangle$, as requested.
\end{proof}

{}From Lemma \ref{circle_test}, we can obtain the following result for the Permutation Test. 

\begin{proposition}\label{perm}
Let $n$ be any number at least $2$. The Permutation Test solves the problem $\mathrm{QSI}_n$ with completeness error probability $0$ and soundness error probability at most $1/n$. 
\end{proposition}

\begin{proof}
Consider a direct application of the Permutation Test. 
Obviously, the Permutation Test has completeness error probability $0$ due to expression (\ref{perm_formula}) of Lemma \ref{circle_test}. Let us fix an arbitrary  NO instance $(|\psi_1\rangle,\ldots,|\psi_n\rangle)$. 
We now argue that in the worst-case scenario, it suffices to consider the case where all indices of our NO instance are divided into two sets $I_1$ and $I_2$ satisfying the following ``equivalence'' conditions: (i) all states whose indices are in $I_1$ (resp. $I_2$) are identical and (ii) any state having an index in $I_1$ and any state having an index in $I_2$ are mutually orthogonal. 
To see that this is sufficient, consider the case where all the indices are divided into three (or more) sets, say,  
$I_1$, $I_2$ and $I_3$. A key observation is that the soundness error probability on the NO instance is at most the soundness error probability on the same instance whose indices are divided into two sets, $I_1$ and $I_2\cup I_3$. Therefore, we need to consider only two sets $I_1$ and $I_2$. 

Now, assume that we have the aforementioned two sets $I_1$ and $I_2$ with 
$|I_1|=l$ and $|I_2|=n-l$ for a certain number $l$ with $1\leq l\leq n-1$. 
For any permutation $\sigma_k$, the value $\prod_{m=1}^{n} \langle\psi_m|\psi_{\sigma_k(m)}\rangle$ becomes $1$ if and only if $\sigma_k$ {\em setwisely stabilizes} $I_1$ and $I_2$; namely, $\sigma_k$ maps any element with an index in $I_1$ (resp. $I_2$) to another element in $I_1$ (resp. $I_2$). 
This property concludes that the soundness error probability of the NO instance equals the ratio between the number of all such permutations and the total number of permutations in $S_n$. This ratio is clearly $l!(n-l)!/n!\leq 1/n$. 
\end{proof}

Under the one-sided error requirement, we can show the optimality of the Permutation Test for $\mathrm{QSI}_n$; namely, any one-sided error quantum operation for $\mathrm{QSI}_n$ must have the soundness error probability of at least $1/n$. Earlier, Kobayashi, Matsumoto, and Yamakami \cite{KMY01} (see also \cite{Bea04}) implicitly proved the optimality of the Permutation Test for $\mathrm{QSI}_2$ (equivalently, the Swap Test). 

\begin{proposition}\label{perm_opt}
Let $n$ be any number greater than $1$. Any quantum operation to solve $\mathrm{QSI}_n$ 
under the one-sided error requirement has soundness probability at least $1/n$.
\end{proposition}

\begin{proof}
Our proof generalizes the new optimality proof for the Swap Test of Hotta and Ozawa \cite{HO07}, whose fundamental idea is similar to \cite{KMY01,Bea04}. Let ${\cal H}$ be our state space. Let $\{E_y,E_n\}$ denote any optimal binary positive operator-valued measure (POVM) that meets the one-sided error requirement, from which we have  
$
E_y(|\psi\rangle^{\otimes n})=|\psi\rangle^{\otimes n}
$
for any state $|\psi\rangle\in{\cal H}$. Let $P_S$ be the projection onto the {\em symmetric subspace} \cite{Bar97}
\[
\hs{-21}
\{ |S_{\mu}\rangle \}=\left\{\left. \sum_{\sigma\in S_n}|m_{\sigma(1)}\rangle\cdots|m_{\sigma(n)}\rangle \right| 
\begin{array}{l}
m_1,m_2,\ldots,m_n\mbox{ are the indices of elements  }\\
\mbox{in the computational basis of ${\cal H}$} 
\end{array}\right\},
\]
which is the subspace of ${\cal H}^{\otimes n}$ that is symmetric under the interchange of states for any pair of positions in the tensor product. Here, we claim that $P_S$ satisfies the equation 
$ E_yP_S=P_S$. 
This claim is shown as follows. Notice that the symmetric subspace is also the subspace of ${\cal H}^{\otimes n}$ spanned by all states of the form $|\psi\rangle^{\otimes n}$ \cite{Bar97}. Using this fact, for any state $|\phi\rangle\in{\cal H}^{\otimes n}$, $P_S|\phi\rangle$ can be expressed as $P_S|\phi\rangle=\sum_\alpha c_\alpha|\varphi_\alpha\rangle^{\otimes n}$. 
The equality $E_y(|\psi\rangle^{\otimes n})=|\psi\rangle^{\otimes n}$ implies that 
$P_S|\phi\rangle$ can be further written as 
\[
\sum_\alpha c_\alpha|\varphi_\alpha\rangle^{\otimes n} 
= \sum_\alpha c_\alpha E_y\left(|\varphi_\alpha\rangle^{\otimes n}\right) 
= E_y\left(\sum_\alpha c_\alpha|\varphi_\alpha\rangle^{\otimes n} \right)\\
= E_yP_S|\phi\rangle.
\]
It follows from the equality $E_yP_S=P_S$ that $E_y=P_S+\sum_\nu\lambda_\nu|A_\nu\rangle\langle A_\nu|$, 
where $\lambda_\nu$ is nonnegative (because $E_y$ is positive) and 
$|A_\nu\rangle$ lies in the orthogonal complement of 
the symmetric subspace. 
Therefore, we conclude that $E_y\geq P_S$. 
Note that the soundness error probability $p_e$ equals
\[
p_e=\mathrm{Tr}[E_y(|\psi_1\rangle\cdots|\psi_n\rangle\langle\psi_1|\cdots\langle\psi_n|)]
\]
for a certain NO instance $(|\psi_1\rangle,\ldots,|\psi_n\rangle)$. We want to show that $p_e\geq 1/n$. 
Now, let us consider a specific NO instance $(|\psi_1\rangle,\ldots,|\psi_n\rangle)$ satisfying that $|\psi_2\rangle=\cdots=|\psi_n\rangle$ as the worst-case instance. {}From the inequality $E_y\geq P_S$, $p_e$ is lower-bounded by
\begin{eqnarray*} 
p_e\geq \mathrm{Tr}[P_S(|\psi_1\rangle\cdots|\psi_n\rangle\langle\psi_1|\cdots\langle\psi_n|)]
=\frac{1}{n!}\sum_{\sigma\in S_n}\prod_{i=1}^n|\langle\psi_i|\psi_{\sigma(i)}\rangle|^2 = \frac{1}{n}.
\end{eqnarray*}
This completes the proof.
\end{proof}

Propositions \ref{perm} and \ref{perm_opt} show the optimality of the Permutation Test for an arbitrary input size $n$.
As for the Circle Test, when $n=3$, we can show in the following proposition that the Circle Test is also optimal under the one-sided error requirement. 

\begin{proposition}\label{circ_3}
The problem $\mathrm{QSI}_3$ is solved with one-sided error probability by the Circle Test with soundness error probability exactly $1/3$. 
\end{proposition}

This proposition follows from a more general statement. For technical reasons, 
we define the {\em Alternation Test} by replacing $S_n$ in the definition of the Permutation Test 
with the alternating group $A_n$, which is the group generated by the even permutations in $S_n$. 

\begin{lemma}\label{alternation-test}
For any number $n\geq2$, the Alternation Test solves the problem $\mathrm{QSI}_n$ with completeness error probability $0$ and soundness error probability at most $1/n$. 
\end{lemma}

Proposition \ref{circ_3} follows immediately from this lemma since $A_3$ equals the cyclic group $C_3$, which defines 
the Circle Test over three states. 

\begin{proofof}{Lemma \ref{alternation-test}}
The Alternation Test has completeness error probability $0$ since, similar to Proposition \ref{perm}, 
the probability $p$ that the Alternation Test outputs EQUAL on $n$ input states $(|\psi_1\rangle,\ldots,|\psi_n\rangle)$ is
$
p = \frac{2}{n!}\sum_{k=0}^{n!/2-1}\prod_{m=1}^{n} \langle\psi_m|\psi_{\tau_k(m)}\rangle,
$
where $\tau_1,\ldots,\tau_{n!/2}$ denote all the even permutations over $[n]$ (in a certain fixed order). 
Hereafter, let us fix a NO instance $(|\psi_1\rangle,\ldots,|\psi_n\rangle)$. 
Similar to the proof of Proposition \ref{perm}, it suffices to deal with the case 
where all indices of this NO instance are divided into two sets $I_1$ and $I_2$ satisfying: (i) all states having indices in $I_1$ (resp. $I_2$) are identical and (ii) any state with an index in $I_1$ and any state with an index in $I_2$ are mutually orthogonal. Assume that $I_1$ and $I_2$ satisfy $|I_1|=l$ and $|I_2|=n-l$ for a certain number $l$ 
with $1\leq l\leq n-1$. 

Note that, for any even permutation $\tau_k$, the value $\prod_{m=1}^{n} \langle\psi_m|\psi_{\tau_k(m)}\rangle$ 
equals $1$ if and only if $\tau_k$ setwisely stabilizes $I_1$ and $I_2$. Thus, the soundness error probability of the NO instance equals the ratio between the number $L$ of all even permutations that setwisely stabilize $I_1$ and $I_2$, and the total number $|A_n|$. 
We will show that this ratio $L/|A_n|$ is exactly $l!(n-l)!/n!$, and hence $L/|A_n|=l!(n-l)!/n! \leq 1/n$. 
Since $|A_n| = n!/2$, it is enough to prove that $L=l!(n-l)!/2$. 
To evaluate $L$, we consider the following two cases: (i) $l=1$ or $l=n-1$ and (ii) $2\leq l\leq n-2$. 

We consider Case (i) when $l=1$. In this case, any even permutation that setwisely stabilizes $I_1$ and $I_2$ 
must fix a unique element in $I_1$, and thus 
it is also an even permutation on $I_2$. 
This implies that $L = |A_{n-1}|$, which is $(n-1)!/2$, as desired. 
In Case (ii), any even permutation that setwisely stabilizes $I_1$ and $I_2$ is either (a) the product of an even permutation over $I_1$ and an even permutation over $I_2$ or 
(b) the product of an odd permutation over $I_1$ and an odd permutation over $I_2$. 
First, we consider the Case (a). Let us consider the total number of products of even permutations over $I_1$ 
and even permutations over $I_2$. This number clearly equals $(l!/2)((n-l)!/2)=l!(n-l)!/4$, which implies that $L=l!(n-l)!/2$.  Case (b) is similar. 
\end{proofof}

Unfortunately, the Circle Test cannot be optimal for certain input sizes $n$. For instance, if $n=4$, the Circle Test can achieve an optimal soundness error probability of $1/4$ for a NO instance $(|\psi\rangle,|\psi\rangle,|\psi\rangle,|\psi^{\bot}\rangle)$, where $|\psi\rangle$ is an arbitrary state in ${\cal H}$ and $|\psi^{\bot}\rangle$ denotes a state orthogonal to $|\psi\rangle$, whereas another NO instance $(|\psi\rangle,|\psi^\bot\rangle,|\psi\rangle,|\psi^\bot\rangle)$ makes the Circle Test produce a soundness error probability of $1/2$ (which is far greater than $1/4$). 

In Section \ref{sec4}, we show that the Circle Test for $\mathrm{QSI}_n$ works asymptotically as good as the Permutation Test, if we incorporate additional classical preprocessing with the Circle Test.

\section{\label{sec3}Approximation of the Circle Test by the Swap Test}

We have shown in the previous section that the Permutation Test and the Circle Test are optimal quantum operations to solve the identity problem $\mathrm{QSI}_3$ with one-sided error probability. {}From a practical viewpoint, it would be ideal to build an identity test for $\mathrm{QSI}_3$ only with the Swap Test as a main quantum ingredient. This is mainly because 
the Swap Test is much simpler than the other two operations, and, more importantly, the Swap Test has been well-studied for its theoretical applications as well as its physical implementations ({\it e.g.}, see \cite{Du06,Hor05,WD07}). How can we develop such a test? A simple and natural approach is a sequential application of the Swap Test (which we refer to as a {\em Swap protocol}). More precisely, a Swap protocol ``classically'' chooses two quantum states for the Swap Test (out of three or more states) and applies the Swap Test to them as its only true ``quantum'' operation. In the following theorem, we present a certain Swap protocol for $\mathrm{QSI}_3$, which asymptotically achieves the same soundness error probability as the Circle Test does. 

\begin{theorem}\label{swap_3states}
Let $m$ be any positive number at least $2$. There exists a Swap protocol for $\mathrm{QSI}_3$, which achieves the soundness error probability of at most $1/3 + 1/4^{m-1}$ by applying the Swap Test $m$ times sequentially. 
\end{theorem}

\begin{proof}
Let $m\geq 2$ be fixed throughout this proof. Our desired Swap protocol for $\mathrm{QSI}_3$, referred to as SRS (Sequential Random Swap), is given as follows.

\medskip
{\bf Protocol SRS($m$)} 

Input: three quantum states $(|\psi_1\rangle,|\psi_2\rangle,|\psi_3\rangle)\in{\cal H}^{\otimes 3}$
\begin{enumerate}
\item Randomly choose two of the three states $|\psi_1\rangle,|\psi_2\rangle$ and $|\psi_3\rangle$. 

\item Repeat the following two steps $m$ times as long as the protocol does not halt. 

\hs{1}  (ii-1) Perform the Swap Test on the chosen two states. If the test outputs NOT EQUAL, 
output NO and halt.

\hs{1}  (ii-2) Choose the leftover state as well as one of the two resulting states at random.

\item Output YES.  
\end{enumerate}

If three input states are identical, then SRS($m$) obviously outputs YES with certainty. Consider the case where all the three input states are mutually orthogonal. Hereafter, we deal only with an arbitrary NO instance $(|\psi_1\rangle,|\psi_2\rangle,|\psi_3\rangle)\in{\cal H}^{\otimes 3}$. 
We first analyze the soundness error probability for $m=2$. In the protocol SRS($2$), the first Swap Test at Step (ii-1) outputs NO with probability exactly $1/2$, regardless of which states are chosen at Step (i). Without loss of generality, we assume that $|\psi_1\rangle$ and $|\psi_2\rangle$ are the chosen states 
at Step (i). After the first Swap Test outputs EQUAL at Step (ii-1), the resulting state is of the form  
$\frac{1}{\sqrt{2}}(|\psi_1\rangle|\psi_2\rangle+|\psi_2\rangle|\psi_1\rangle)$ since $|\psi_1\rangle$ and $|\psi_2\rangle$ are orthogonal. At Step (ii-2), we obtain two input states: the pure state $|\psi_3\rangle$ and the mixed state $\pmb{\rho}=\frac{1}{2}(|\psi_1\rangle\langle\psi_1|+|\psi_2\rangle\langle\psi_2|)$. These input states can be evaluated 
as EQUAL by the second Swap Test at Step (ii-1) with probability exactly 
$\frac{1}{2}+\frac{1}{2}\mathrm{Tr}(\pmb{\rho}|\psi_3\rangle\langle\psi_3|)=\frac{1}{2}$. Therefore, we obtain the correct answer NO at Step (ii-1) with probability exactly $\frac{1}{2}+\frac{1}{2}\cdot\frac{1}{2} = \frac{3}{4}$. This gives the soundness error probability of $1/4$, which is smaller than $1/3$. Since the soundness error probability of SRS($m$) decreases as $m$ becomes larger, we conclude that SRS($m$) has soundness error probability smaller than $1/3$.

The more complex case is that two input states are identical and the rest is orthogonal to them. Because of the symmetry of our protocol, we can assume that $|\psi_1\rangle=|\psi_3\rangle$ and $|\psi_2\rangle=|\psi_1^\bot\rangle$. We need to consider the following two cases.
\begin{itemize}
\item[(a)] $|\psi_1\rangle$ and $|\psi_2\rangle$ (or alternatively $|\psi_2\rangle$ and $|\psi_3\rangle$) are chosen at 
Step (i). 

\item[(b)] $|\psi_1\rangle$ and $|\psi_3\rangle$ are chosen at Step (i). 
\end{itemize}

We begin with Case (a). For notational convenience, we use the following abbreviations: 
$|1\rangle:=|\psi_1^\bot\rangle|\psi_1\rangle|\psi_1\rangle$, 
$|2\rangle:=|\psi_1\rangle|\psi_1^\bot\rangle|\psi_1\rangle$ and 
$|3\rangle:=|\psi_1\rangle|\psi_1\rangle|\psi_1^\bot\rangle$. 
It is not important for us to choose, at Step (ii-2), which of the two resulting states to apply the Swap Test, since if the protocol does not halt, after Step (ii-1), on, say,  the first and the second states, the obtained state is in the form: 
$\alpha(|1\rangle+|2\rangle)+\beta|3\rangle$. For simplicity, we assume that the second state is always chosen at Step (ii-2). 
For our further analysis, we need the following lemma. For readability, we ignore normalization factors of quantum states in the lemma. 

\begin{lemma}\label{lemma_pk}
Let $k$ be any number in $[m]$. Under the condition that the protocol does not halt after the $(k-1)$th Swap Test in Case (a), the (conditional) probability $p_k$ that the protocol does not halt after the $k$th Swap Test is $p_{k}=1-\frac{6}{4^{k}+8}$. The obtained (non-normalized) state can be represented as $(a_k+1)|1\rangle+(a_k+1)|2\rangle+a_k|3\rangle$, 
where $a_k=\frac{2}{3}(4^{(k-1)/2}-1)$, if $k$ is odd, and $a_k|1\rangle+(a_k+1)|2\rangle+(a_k+1)|3\rangle$, 
where $a_k=\frac{1}{3}(4^{k/2}-1)$, if $k$ is even. 
\end{lemma}

Meanwhile, we postpone the proof of this lemma. Let $q_k$ be the (accumulative) probability that the protocol does not halt 
after the $k$th Swap Test in Case (a). Since $q_1=p_1$ and $q_k=p_kq_{k-1}$ for any $k\geq 2$, Lemma \ref{lemma_pk} implies that $q_1=1/2$ and $q_k=\left(1-\frac{6}{4^k+8}\right)q_{k-1}$. These recurrence equations have a unique solution $q_k=\frac{1}{3}+\frac{2}{3\cdot 4^k}$ for any number $k\geq1$.  

Next, let us consider Case (b). Let $r_k$ be the (accumulative) probability that the protocol does not halt after the $k$th Swap Test in Case (b). Under our assumption, Case (b) can be analyzed in the same way as Case (a) if we replace $k$ in Case (a) by $k+1$, because the first Swap Test makes no effect on its subsequent computation. We then obtain that $r_1=1$ and $r_k=q_{k-1}$ for any number $k\geq 2$.

Notice that Case (a) holds with probability $2/3$ and Case (b) holds with probability $1/3$. Therefore, if the given input is a NO instance, where two of the three states are identical and the other is orthogonal to them, the protocol SRS($m$) 
outputs YES at Step (iii) with probability $(2/3)q_m+(1/3)r_{m}=1/3+1/4^{m-1}$, as requested. 

\begin{proofof}{Lemma \ref{lemma_pk}}
The proof is done by induction on $k\geq 1$. Let $p_k$ denote the probability that the protocol does not halt after $k$th Swap Test in Case (a). Consider the basis case $k=1$. After the first Swap Test, since we obtain the state $|1\rangle+|2\rangle$ with probability $1/2$, the protocol outputs NO with probability exactly $1/2$. Hence, we have $p_1=1/2$ and $a_1=0$. In the case of $k=2$, note that the Swap Test is applied to the second and third states. The total state including the first register (used by the quantum Fourier transform) evolves by the Swap Test as follows:
\begin{eqnarray*}
(|0\rangle+|1\rangle)(|1\rangle+|2\rangle) 
&\mapsto& |0\rangle (|1\rangle+|2\rangle)+|1\rangle(|1\rangle+|3\rangle)\\
&\mapsto& (|0\rangle+|1\rangle)(|1\rangle+|2\rangle)
+(|0\rangle-|1\rangle)(|1\rangle+|3\rangle)\\
& & =\ |0\rangle(2|1\rangle+|2\rangle+|3\rangle)+|1\rangle(|2\rangle-|3\rangle). 
\end{eqnarray*} 
Provided that the protocol does not halt, we obtain the state $2|1\rangle+|2\rangle+|3\rangle$, which yields $a_2=1$. The desired probability $p_2$ is thus calculated as 
\[
 p_2=1-\frac{1^2+(-1)^2}{2^2+1^2+1^2+1^2+(-1)^2}=3/4.
\] 
This yields the lemma for the case $k=2$. 

Next, let $k$ be any integer greater than $2$. First, we deal with the case where $k$ is odd. Assuming that the lemma holds for $k$, we want to show that the lemma also holds for $k+1$. By our induction hypothesis, we have the state $\ket{\psi_k}= (a_k+1)|1\rangle+(a_k+1)|2\rangle+a_k|3\rangle$ after the $k$th Swap Test, where $a_k=\frac{2}{3}(4^{(k-1)/2}-1)$. 
Note that the $(k+1)$th Swap Test is applied to the second and third states in $\ket{\psi_k}$. 
The Swap Test makes the total state evolve as follows:
\begin{eqnarray*}
& &\!\!\!\!\!\!\!\!\!\!\!\!\!\! (|0\rangle+|1\rangle)((a_k+1)|1\rangle+(a_k+1)|2\rangle+a_k|3\rangle) \\
&\mapsto& |0\rangle ((a_k+1)|1\rangle+(a_k+1)|2\rangle+a_k|3\rangle)\\
&       &   \ \ \ \    +|1\rangle ((a_k+1)|1\rangle +(a_k+1)|3\rangle +a_k|2\rangle)\\
&\mapsto& (|0\rangle+|1\rangle)((a_k+1)|1\rangle+(a_k+1)|2\rangle+a_k|3\rangle)\\
&       &   \ \ \ \   +(|0\rangle-|1\rangle)((a_k+1)|1\rangle+a_k|2\rangle+(a_k+1)|3\rangle)\\
& & =\ |0\rangle((2a_k+2)|1\rangle+(2a_k+1)|2\rangle+(2a_k+1)|3\rangle)
+|1\rangle(|2\rangle-|3\rangle). 
\end{eqnarray*}  
We then obtain the state $(2a_k+2)|1\rangle+(2a_k+1)|2\rangle+(2a_k+1)|3\rangle$ if the protocol does not halt. From this state, it immediately follows that $a_{k+1}=2a_k+1$. Therefore, $p_{k+1}$ has the value 
\begin{eqnarray*}
p_{k+1} &=& 1-\frac{1^2+(-1)^2}{(2a_k+2)^2+2(2a_k+1)^2+1^2+(-1)^2}\\
        &=& 1-\frac{2}{(a_{k+1}+1)^2+2a_{k+1}^2+2}.
\end{eqnarray*}
Since $a_{k+1}=2a_k+1 =\frac{1}{3}(4^{(k+1)/2}-1)$, we finally obtain $p_{k+1} = 1-\frac{6}{4^{k+1}+8}$. We thus conclude, from the induction hypothesis for $k$, that the lemma holds for $k+1$. A similar analysis verifies that the induction step also holds for any even number $k$. Therefore, the mathematical induction guarantees the correctness of the lemma. 
\end{proofof}

This completes the proof of the theorem.
\end{proof}

As a direct consequence of Theorem \ref{swap_3states}, we conclude that SRS is one of the best choices among all Swap protocols solving the problem  $\mathrm{QSI}_3$.

\section{\label{sec4}Approximation of the Permutation Test by the Circle Test}

This section compares the performances of the Circle Test and of 
the Permutation Test. First, we focus our attention on the Circle Test for $\mathrm{QSI}_n$, where  $n$ is a prime number. For such a number $n$, we can show that the Circle Test has the same performance for $\mathrm{QSI}_n$ as the Permutation Test does. This indicates that the Circle Test is 
a best quantum test for any ``prime'' input size $n$ among all one-sided error quantum operations for $\mathrm{QSI}_n$. 

\begin{proposition}
Let $n$ be a prime number. The Circle Test for $\mathrm{QSI}_n$ achieves the soundness error probability of at most $1/n$. 
\end{proposition}

\begin{proof}
Let $n$ be any prime number and let $(|\psi_1\rangle,\ldots,|\psi_n\rangle)$ be any instance of the identity problem $\mathrm{QSI}_n$. Lemma \ref{circle_test} implies that the Circle Test outputs EQUAL on the instance 
with the probability   
$
p =  \frac{1}{n}\sum_{k=0}^{n-1}\prod_{m=1}^n\langle\psi_m|
\psi_{\sigma_c^k(m)}\rangle.
$ 
If $(|\psi_1\rangle,\ldots,|\psi_n\rangle)$ is an YES instance, then it is straightforward to show that $p=1$. Next, we consider the case where $(|\psi_1\rangle,\ldots,|\psi_n\rangle)$ is a NO instance. Now, we claim the following.

\begin{lemma}\label{lemma_orth}
Let $(|\psi_1\rangle,\ldots,|\psi_n\rangle)$ be any NO instance. 
For any number $k\in[n-1]$, there exists an index $m\in[n]$ such that $\langle\psi_m|\psi_{\sigma_c^k(m)}\rangle=0$. 
\end{lemma}

{}From this lemma, it follows that the probability $p$ equals $\frac{1}{n}\prod_{m=1}^n\langle\psi_m|\psi_{\sigma_c^0(m)}\rangle$. Therefore, the Circle Test outputs EQUAL on $(|\psi_1\rangle,\ldots,|\psi_n\rangle)$ with probability  
$\frac{1}{n}\prod_{m=1}^n\langle\psi_m|\psi_{\sigma_c^0(m)}\rangle$, which is clearly upper-bounded by $1/n$.

To complete the proof of the proposition, we need to prove Lemma \ref{lemma_orth}. Let us assume, toward a contradiction, that the lemma fails. By the promise of $\mathrm{QSI}_n$, there exists a number $k\in[n-1]$ such that, for any number $m\in[n]$, 
$
|\psi_m\rangle=|\psi_{\sigma_c^k(m)}\rangle.
$ 
Since $(|\psi_1\rangle,\ldots,|\psi_n\rangle)$ is a NO instance, there exist two indices $\mu$ and $\mu'$ for which 
$\langle\psi_\mu|\psi_{\mu'}\rangle=0$. This yields the existence of a proper subset $I=\{{\mu_1},{\mu_2},\ldots\}$ of $[n]$ satisfying that $|\psi_\mu\rangle=|\psi_\nu\rangle$  for any pair $\mu,\nu\in I$, and  $\langle\psi_\mu|\psi_\nu\rangle=0$ for any $\mu\in I$ and $\nu\in [n]\setminus I$. Choose $\mu_1$ in $I$. 
Since $|\psi_m\rangle=|\psi_{\sigma_c^k(m)}\rangle$ for any $m\in[n]$, we obtain 
$
|\psi_{\mu_1}\rangle=|\psi_{\sigma_c^k(\mu_1)}\rangle=|\psi_{\sigma_c^{2k}(\mu_1)}\rangle=\cdots.
$ 
Let $S=\{\mu_1,\sigma_c^k(\mu_1),\sigma^{2k}_c(\mu_1),\ldots\}$. It follows from the definition of $I$ that $S \subseteq I$. Since the set $S$ is the $\mathbb{Z}_n$-orbit with respect to $\mu_1$, its cardinality is a divisor of $n$. The ``prime'' condition of $n$ concludes that $[n] = S$. Since $S\subseteq I$, we have $I=[n]$, which contradicts our assumption that $I$ is a proper subset of $[n]$. This completes the proof of the lemma and thus completes the proof of the proposition.
\end{proof}

For an arbitrary input size $n$, how good is the performance of the Circle Test for $\mathrm{QSI}_n$, compared to the Permutation Test? Under the one-sided error requirement, as seen in Section \ref{sec2}, the Circle Test,  in general, cannot be optimal for $\mathrm{QSI}_n$. Nevertheless, it is possible to give a simple and almost optimal protocol, called RCIR (Randomized Circle Test), which uses the Circle Test only once after the classical processing of permuting $n$ quantum states randomly. 

\medskip
{\bf Protocol RCIR} 

Input: $n$ quantum states $(|\psi_1\rangle,|\psi_2\rangle,\ldots,|\psi_n\rangle)\in{\cal H}^{\otimes n}$
\begin{enumerate}
\item Permute the input quantum states by a randomly chosen permutation $\tau\in S_n$. Let   $(|\phi_1\rangle,\ldots,|\phi_n\rangle)$, where $|\phi_j\rangle=|\psi_{\tau(j)}\rangle$, be the resulting quantum states. 

\item Apply the Circle Test to $(|\phi_1\rangle,\ldots,|\phi_n\rangle)$. 
\end{enumerate}

We show that the protocol RCIR is an ``asymptotically'' optimal quantum operation for $\mathrm{QSI}_n$ up to a constant multiplicative factor of nearly $\pi^2/6$.

\begin{theorem}\label{general_n}
The protocol RCIR meets the one-sided error requirement and achieves the soundness error probability of 
at most $\pi^2/6n+O(1/n^2)\leq 1.7/n+O(1/n^2)$.
\end{theorem}

\begin{proof}
With the same reasoning given in the proof of Proposition \ref{perm}, it suffices to analyze only NO instances $(|\psi_1\rangle,\ldots,|\psi_n\rangle)$ whose indices are divided into two sets  $I_1$ and $I_2$ such that any two states with indices in $I_1$ (also, $I_2$) are identical and any pair of states, one of which has an index in $I_1$ and the other has an index in $I_2$, is orthogonal. In what follows, we call a state whose index is in $I_1$ (resp. $I_2$) an {\em $I_1$-state} (resp. {\em $I_2$-state}). Let $I_1$ and $I_2$ be such sets of indices of the permuted states $\{|\phi_1\rangle,\ldots,|\phi_n\rangle\}$ obtained at Step (i) of the protocol RCIR. For convenience, let $I_1=\{\mu_1,\mu_2,\ldots,\mu_r\}$ with $\mu_1<\mu_2<\cdots<\mu_r$ and $I_2=[n]\setminus I_1$, where $r\in [n-1]$. 
Without loss of generality, we assume that $|I_1|\leq |I_2|$; namely, $r\leq n/2$. Let us also assume that there are exactly $s$ elements $k_1=0,k_2,\ldots,k_s\in\{0,1,\ldots,n-1\}$ such that each number $k\in\{k_1,\ldots,k_s\}$ satisfies $\langle\phi_m|\phi_{\sigma_c^k(m)}\rangle =1$ for any number $m\in[n]$. Lemma \ref{circle_test} concludes that the soundness error probability of the protocol equals $s/n$. The following lemma is easily proven. 

\begin{lemma}\label{lemma_easy}
Let $K=\{k_1,k_2,\ldots,k_s\}$. 

(i) For any $m\geq 1$, if $k'\in K$ then so is $mk'$. 

(ii) If $k',k''\in K$ then so is $\mathrm{GCD}(k',k'')$.    
\end{lemma}

By Lemma \ref{lemma_easy}, the set $K=\{k_1,k_2,k_3,\ldots,k_s\}$ can be of the form $k_1=0,k_2=k,k_3=2k,\ldots,k_s=(s-1)k$ for the divisor $k$ $(=n/s)$ of $n$. For convenience, we call $s$ and $k$ the {\em repetition number} and the {\em cycle size}, respectively.  
 
\begin{figure}[t]
\begin{center}
\includegraphics*[width=4.0cm]{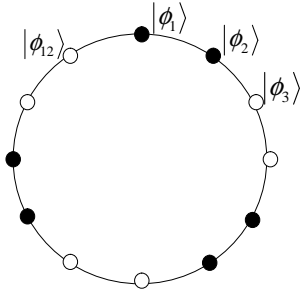}
\caption{An example of cyclic alignment of all $I_1$-states with $n=12$ and $|I_1|=6$}\label{figure1} 
\end{center}
\end{figure}

To help the reader, let us see an example. Figure \ref{figure1} renders a cyclic alignment of all $I_1$-states with parameters $n=12$ and $r=6$, where the repetition number $s$ is $3$ and the cycle size $k$ is $4$. A black node (resp.\ white node) indicates an $I_1$-state (resp.\ $I_2$-state). Each cyclic alignment of $I_1$-states induces an {\em $I_1$-pattern}, 
which is a bit string $(b_1,\ldots,b_{k})$ defined by $b_i=1$ if $i\in I_1$ and $0$ otherwise. In Figure \ref{figure1}, this $I_1$-pattern is $(1,1,0,0)$. By the definition of cycle size $k$, such an $I_1$-pattern uniquely characterizes a cyclic alignment of $I_1$-states as follows: for any $j\in\{0,1,\ldots,s-1\}$ and $i\in[k]$, the index $jk+i\in [n]$ is in $I_1$ if $i\in I_1$, and in $I_2$ if $i\in I_2$. Note that the Hamming weight of $(b_1,\ldots,b_k)$, which indicates the number of indices in $[k]\cap I_1$, is exactly $|I_1|/s=r/s$.

Now, we return to our proof. We wish to show that the soundness error probability for any fixed $r$ is at most $\pi^2/6n+O(1/n^2)$. Let $p_{s}$ be the probability that a cyclic alignment of $I_1$-states with repetition number $s$ is chosen by the protocol RCIR. Note that, as far as $\mathrm{GCD}(n,r)=1$, $s$ equals $1$; hence, we have $p_1=1$. This implies that the soundness error probability equals $1/n$. Since a cyclic alignment of all $I_1$-states is randomly chosen, it follows that, for $s\geq 2$, $p_s\leq \frac{{k\choose r/s}}{{n\choose r}}=\frac{{n/s\choose r/s}}{{n\choose r}}$. Recall that any cyclic alignment of $I_1$-states with repetition number $s$ produces the soundness error probability of $s/n$. Therefore, the total soundness error probability is at most
\begin{equation}\label{eq_upper}
p_1\cdot\frac{1}{n}+\sum_{s:\ s|n}\frac{{n/s\choose r/s}}{{n\choose r}}\cdot\frac{s}{n}
\leq \frac{1}{n}+\sum_{s:\ s|n}\frac{{n/s\choose r/s}}{{n\choose r}}\cdot\frac{s}{n}.
\end{equation}

To upper-bound Eq.(\ref{eq_upper}) further, we need the following technical lemma. Recall that $s$ is a divisor of $r$. For convenience, let $q(n,r,s)=\frac{{n/s\choose r/s}}{{n\choose r}}\cdot\frac{s}{n}$.

\begin{lemma}\label{lemma_eva}
The value $q(n,r,s)$ is at most $\frac{1}{ns^2}$ if $s\leq r/3$; $\frac{6}{(n-1)(n-2)(n-3)}$ if $s=r/2$; and $\frac{2}{n(n-1)}$ if $s=r$. 
\end{lemma}

We continue our argument. Lemma \ref{lemma_eva} helps us upper-bound the right-hand expression in Eq.(\ref{eq_upper}) as
\begin{eqnarray*}
\lefteqn{\frac{1}{n}+\sum_{s:\ s|n}\left(\frac{1}{ns^2}\right)+O(1/n^2)} \hs{5} \\
&\leq& \frac{1}{n}+\sum_{s=2}^\infty \left(\frac{1}{ns^2}\right)+O(1/n^2) 
\;\;=\;\; \frac{1}{n}+\frac{\pi^2/6-1}{n}+O(1/n^2).
\end{eqnarray*}
The last expression is clearly equal to $\pi^2/6n +O(1/n^2)$, as requested. 

What remains is to prove Lemma \ref{lemma_eva}. Consider the first case $s=r$. In this case, we have $q(n,r,s)
=\frac{{n/s \choose 1}}{{n\choose s}}\cdot\frac{s}{n}=\frac{1}{{n\choose s}}$, which is bounded from above by 
$\frac{1}{{n\choose 2}}=\frac{2}{n(n-1)}$ because $s\geq 2$. Let us consider the second case $s=r/2$. 
The expression $q(n,r,s)$ is further calculated as
\begin{equation}\label{eq_r2s}
q(n,r,s) = \frac{{n/s\choose 2}}{{n\choose 2s}}\cdot\frac{s}{n}
=\frac{2s(2s-1)\cdots 1}{n(n-1)\cdots (n-2s+1)} \cdot \frac{n/s-1}{2}. 
\end{equation}
Noting that $2s\geq 4$, $2s=r\leq n/2$ and $n-2s+1\geq n/2+1$, it follows from Eq.(\ref{eq_r2s}) that $q(n,r,s)$ 
is at most
\begin{eqnarray*}
\lefteqn{\frac{1\cdot 2\cdot 3\cdot 4\cdots 2s}{n(n-1)(n-2)(n-3)\cdots (n-2s+1)}\cdot\frac{n}{2s}}\hs{5}\\
&\leq& \frac{1\cdot 2\cdot 3\cdot 4}{n(n-1)(n-2)(n-3)}\cdot\frac{n}{4} =\frac{6}{(n-1)(n-2)(n-3)}.
\end{eqnarray*}
In the final case $s\leq r/3$, the expression $q(n,r,s)$ equals
\begin{eqnarray*}
q(n,r,s) &=& \frac{{n/s\choose r/s}}{{n\choose r}}\cdot\frac{s}{n}
=\frac{\frac{(n/s)(n/s-1)\cdots(n/s-r/s+1)}{(r/s)(r/s-1)\cdots 1}}{\frac{n(n-1)\cdots (n-r+1)}{r(r-1)\cdots 1}}
\cdot\frac{s}{n}\\
&=& \frac{s\cdot\frac{n}{s}\cdots\left(\frac{n}{s}-\frac{r}{s}+1\right)\cdot r\cdots 1}{n\cdot n\cdots 
\left(n-r+1\right)(\frac{r}{s}\cdots 1)},
\end{eqnarray*}
which is clearly at most
\begin{eqnarray*} 
\frac{s}{n}\left(\frac{n/s}{n}\right)^{r/s}\frac{r\cdots(\frac{r}{s}+1)}{(n-\frac{r}{s})\cdots(n-r+1)}.
\end{eqnarray*}
This expression is further upper-bounded by 
$\frac{s}{n}(\frac{1}{s})^{r/s}$, 
since $2\leq s\leq r\leq n/2$ implies 
$\frac{r}{n-\frac{r}{s}}\leq 2/3$. 
{}From our assumption $s\leq r/3$, it therefore follows that  
\[
q(n,r,s)\leq\frac{s}{n}\left(\frac{1}{s}\right)^{r/s}\leq\frac{s}{n}\left(\frac{1}{s}\right)^3=\frac{1}{ns^2}.
\]
This ends the proof of the lemma and thus the proof of Theorem \ref{general_n}.
\end{proof}

\section{\label{sec5}Closing Discussion}

The Swap Test has been widely used in the literature to test the identity of two quantum states. In this paper, we have studied two additional tests, the Permutation Test and the Circle Test, which generalize the Swap Test. We have analyzed the performances of these two tests for the quantum state identity problem, $\mathrm{QSI}_n$, under the one-sided error requirement. Throughout this paper, we have restricted our attention to the identity problem's {\em promise} (in the definition of $\mathrm{QSI}_n$) and also the {\em one-sided error requirement}. These restrictions make our analysis easier; nevertheless, the restrictions can be relaxed. We briefly discuss how our result can be applied to less constrained situations. 
 
The promise of our identity problem $\mathrm{QSI}_n$ demands that any pair of quantum states is identical or orthogonal. 
By relaxing the latter orthogonality, we can consider the following weak form of an identity problem, denoted $\mathrm{QSI}_n^\varepsilon$, in which we want to determine either (a) all $n$ quantum states are identical or (b) there are two states whose inner product is less than or equal to $\varepsilon$, provided that either (a) or (b) holds. This problem $\mathrm{QSI}_n^\varepsilon$ was dealt with in a fingerprinting protocol 
in \cite{BCWW01}. Our results in this paper still provide a good proximity of the three tests to the problem $\mathrm{QSI}_n^\varepsilon$ since $\mathrm{QSI}_n$ coincides with $\mathrm{QSI}_n^\varepsilon$ when $\varepsilon=0$.

Our one-sided error requirement requests that the completeness error probability should be $0$. This requirement naturally occurs in the literature regarding the Swap Test ({\it e.g.}, \cite{BCWW01,Bea04,SWS07}). As a natural relaxation of this requirement, when we allow non-zero completeness error probability, we obtain the {\em two-sided error requirement}. Even with this relaxed requirement, we can claim that the Swap Test is ``optimal'' in the sense that the Swap Test achieves the largest gap between the probabilities that EQUAL is outputted on YES instances and on NO instances. 
This claim can be shown by a {\em trace-norm distance argument} as follows. 

Consider the two YES instances $(|0\rangle,|0\rangle)$ and $(|1\rangle,|1\rangle)$,  
where each input is a single qubit. For simplicity, let us denote them by $|00\rangle$ and $|11\rangle$, respectively. 
Similarly, consider two NO instances $|+-\rangle$ and $|-+\rangle$, where 
$|+\rangle=(|0\rangle+|1\rangle)/\sqrt{2}$ and $|-\rangle=(|0\rangle-|1\rangle)/\sqrt{2}$. 
Now, let $\pmb{\rho}_y=\frac{1}{2}(|00\rangle\langle00|+|11\rangle\langle 11|)$ 
and $\pmb{\rho}_n=\frac{1}{2}(|+-\rangle\langle +-|+|-+\rangle\langle -+|)$. 
Write $p_c$ and $p_s$ for the completeness and soundness error probabilities,  respectively, of the test.
There is a POVM $M$ such that the $l_1$-norm gap $GAP$ between two probability distributions 
obtained by $M$ on $\pmb{\rho}_y$ and $\pmb{\rho}_n$ is at least $|(1-p_c)-p_s|+|(1-p_s)-p_c|=2-2p_c-2p_s$. 
On the contrary, since the trace-norm distance between $\pmb{\rho}_y$ and $\pmb{\rho}_n$ is $1/2$, 
the value $GAP$ should be at most $1$ \cite{ANK98} 
(see also \cite{KMY03}). This yields the inequality  
$p_c+p_s\geq 1/2$. Notice that the Swap Test achieves the equality $p_c+p_s=1/2$. Therefore, the Swap Test is optimal even in the two-sided error requirement. 

With a similar argument for the Circle Test for $\mathrm{QSI}_3$, 
we can prove that the Circle Test is also ``optimal'' with two-sided error probability. 
On the contrary, the optimality of the Permutation Test under the two-sided error requirement is currently open. We expect the optimality of the Permutation Test; however, it is likely that the trace-norm distance argument  for the Permutation Test 
is insufficient to prove the optimality under the two-sided error requirement.

Another interesting open question in line of our work is to seek an efficient approximation of the Permutation Test for $\mathrm{QSI}_n$ by use of a certain Swap protocol that runs the Swap Test $O(n)$ times. Such a Swap protocol provides an ideal construction of a quantum circuit that implements the Permutation Test since it is much more concise than the direct construction of the Permutation Test based on the decomposition of the Fourier transform $F_{n!}$ over $n!$ elements.


\paragraph{\large Acknowledgments:} 
We are grateful to Masahiro Hotta and Masanao Ozawa for sending us 
their unpublished manuscript that became a basis of our proof of Proposition \ref{perm_opt}. 


\end{document}